\documentclass[11pt]{amsart}

\usepackage{lipsum}

\newcommand\blfootnote[1]{%
  \begingroup
  \renewcommand\thefootnote{}\footnote{#1}%
  \addtocounter{footnote}{-1}%
  \endgroup
}

\usepackage{AFBoix2015}

\begin{document}

\author[A. F. Boix]{Alberto F.\,Boix$^1$}
\thanks{$^1$ Corresponding author address: IMUVA--Mathematics Research Institute, Universidad de Valladolid, Paseo de Belen, s/n, 47011, Valladolid, Spain. Email of the corresponding author: alberto.fernandez.boix@uva.es}

\author[A. Segura Moreiras]{Adri\'an Segura Moreiras$^2$}
\thanks{$^2$ Departament of Law, Universitat Pompeu Fabra, Ramon Trias Fargas 25-27, 08010 Barcelona, SPAIN. \\ adriansegura@hotmail.es}

\title[Determinacy on New Keynesian models]{Revisiting the determinacy on New Keynesian Models: A survey}


\maketitle


\begin{abstract}
The goal of this paper is to review some analytic techniques that are potentially useful to shed light on the determinacy question that
arises in New Keynesian models as result of a combination of
several monetary policy rules; in these models, we provide conditions
to guarantee existence and uniqueness of equilibrium by means
of results that are obtained from theoretical analysis. In particular, these methods confirm the well known fact that Taylor--like rules in interest rate setting are not the
only way to reach determinacy of the rational expectations equilibrium
in the New Keynesian setting. The key technical tool we use
for that purposes is the so--called Budan--Fourier Theorem, that
we review along the paper. All the ideas and techniques presented have been already used, our contribution that might be original here are the organization and emphasis.
\end{abstract}

\section{Introduction}

The indeterminacy  of the rational expectations equilibrium (REE) poses a complication to the conduct of monetary policy. It is associated with increased volatility as there is uncertainty about which equilibrium will be realized. Hence, it is possible that agents in the economy will produce a second-best outcome in equilibrium. This means that a policy regime in place should not only be consistent with an optimal equilibrium, but also concerned about its uniqueness. 

One of the main problems that arises from New Keynesian (NK) models is the so-called multiple equilibria puzzle. This captures the idea that an undesired equilibrium could appear as a result of a specific  combination of policies. Cochrane \cite{Cochrane2011} argued that there have been many attempts to tackle this problem. However, practically all of them seem to assume that the government will have the power to blow up the economy if an unexpected equilibrium occurs.

The discussion among academics is still open with new alternative solutions to the dilemma recently proposed. In Bianchi and Nicol\`o \cite{BianchiNicolo2017}, the authors used a method that consists of augmenting the original model with auxiliary exogenous equations in order to provide the adequate number of explosive roots. Our paper, addressing the same problem, takes a different direction. We pretend to explore, analytically and numerically, the conditions under which uniqueness and existence are guaranteed in equilibrium. Therefore, the main purpose of our work is to shed light on the determinacy question. We are going to compare computational results (from simulation) with those that are obtained from theoretical analysis; however, we want to
single out that we do not produce results based on computational
simulation, our path is always, on the one hand, to exhibit
statements that are obtained from theoretical analysis and, on
the other hand, to illustrate these statements by means of
numerical simulation.

The focus of attention in the NK framework has been centered around the determinacy conditions of endogenous interest rate rules of the kind presented by Taylor in \cite{Taylor1993}. Determinacy there has often been found to depend on the size of the policy response parameters, or more specifically, the Taylor principle being followed (see \cite{Woodford2001} and \cite{BullardMitra2002}). To guarantee determinacy, the literature has also highlighted the importance of historical feedback in monetary rules, with purely forward-looking policy found to foster non-uniqueness of equilibria \cite{SvenssonWoodford2005}. Our results here display the well known fact, already pointed out among other by Woodford, Bullard and Mitra, that “Taylor-like” rules in interest rate setting are not the only way to achieve determinacy of the REE in the NK setting.

The importance of our work stems from the fact that we propose a generalization of computational results (applied  to finding the roots of a characteristic polynomial) in order to clarify the existence, and potentially, uniqueness of real roots for a linear system of equations. To do so, we use the so-called Budan-Fourier Theorem which is stated in the paper (see Theorem \ref{Budan Fourier}); this result has been already used for the same purposes in some earlier works to tackle the determinacy issue (e.g. \cite[Proof of Proposition 1]{BullardMitra2007preprint}).

The paper is organized as follows: after reviewing the Budan--Fourier Theorem in
Section \ref{the budan fourier theorem} and some determinacy terminology
we introduce in Section \ref{determinacy terminology} to shorten and simplify the statement
of several of the results, in Section \ref{Gali model}, we look at a canonical New Keynesian (NK) model and derive the conditions for determinacy of equilibrium when the money supply follows an exogenous path.  In Section \ref{models with lagged data}, we consider a model in which a monetary authority responds to lagged values of output, inflation, and interest rate deviations. In Section \ref{behavioral model}, we explore stability conditions for a model in which agents do not fully understand future policies. Finally, in
Section \ref{potential limitations}, we explore the potential limitations of these methods to
tackle more involved models; even in this case, we illustrate that the
Budan--Fourier Theorem is still a useful tool to provide necessary
determinacy conditions that are easy to check in practice, because they
only require polynomial evaluation. It is well known that, when
the characteristic equation is of degree two, there are several more
elementary ways to tackle this issue; for instance, Chatelain and
Ralf \cite[Proposition 1]{ChatelainRalf2018} use the fact that, when
the characteristic equation is of degree two, the eigenvalues
are non--linear functions of the trace and the determinant of the corresponding
matrix \cite[pages 63--67]{Azariadis1993}.

We want also to mention here in what linear rational
expectational models we are interested to tackle the
determinacy issue in this paper; indeed, all our models can be cast in the
form
\[
\Gamma_0 y(t)=\Gamma_1 y(t-1)+\Psi z(t),
\]
($t=1,\ldots ,T$), where $z(t)$ is an exogenously evolving, possibly
serially correlated, random disturbance, and $\Gamma_0$ is an
invertible matrix. This fits into the framework studied in
\cite{BlanchardKhan1980}, where the determinacy issue boils down
to count how many eigenvalues of $\Gamma_0^{-1}\Gamma_1$ lie inside or
outside the complex unit disk. It is worth noting that both in
\cite{Sims2002} and \cite{LubikSchorfheide2003} the authors tackle
more general models than the ones considered here, in particular they allow
$\Gamma_0$ to be singular; however, in both works they also need to
calculate some eigenvectors of certain matrices, and it is well known that, if
one wants to do so in practice (e.g. numerically), then one often has
to calculate at once both eigenvalues and eigenvectors, for instance
using the classical Power method \cite[Chapter 7]{Householder1975}. For this
reason, we hope that the methods we review in this paper will be useful
to tackle more complicated models. Again, as we already pointed out, what might be original in this manuscript is the organization of the material and the emphasis, hoping that will be potentially useful for researchers working in this subject. The list of references at the end gives an indication of the provenance of the fundamental ideas and techniques, and might suggest directions for additional research.

All our results will be illustrated through numerical examples that were done
with Matlab \cite{MATLAB:2015}.

\section{The Budan--Fourier Theorem}\label{the budan fourier theorem}
Due to the importance that the Budan--Fourier Theorem plays in this
paper, our goal now is to review it for the convenience of the
reader, referring to \cite[Theorem 1]{Akritas1982} and the references
given therein for further details (see also
\cite[page 173]{Barbeau1989}). First of all, we define the notion
of sign variation.

\begin{df}\label{sign variation}
We say that a \textbf{sign variation exists} between two nonzero
numbers $c_p$ and $c_q$ ($p<q$) in a finite or infinite sequence of
real numbers $c_1,c_2,c_3,\ldots$ if the following holds.

\begin{enumerate}[(i)]

\item If $q=p+1,$ then $c_p$ and $c_q$ have opposite signs.

\item If $q\geq p+2,$ then $c_{p+1},\ldots ,c_{q-1}$ are all zero
and $c_p$ and $c_q$ have opposite signs.

\end{enumerate}
\end{df}
Keeping in mind this terminology, the Budan--Fourier Theorem
can be phrased in the below way.

\begin{teo}[Budan--Fourier]\label{Budan Fourier}
Let $P(x)$ be a polynomial with real coefficients and of degree
$d,$ and denote by $P^{(i)}$ its $i$th derivative; moreover, set
$P_{seq} (x):=(P(x),P'(x),P^{(2)}(x),\ldots, P^{(d)} (x)).$ Finally, given
real numbers $a<b,$ denote by $v_a$ (respectively, $v_b$) the number
of sign variations of $P_{seq} (a)$ (respectively, $P_{seq} (b)$). Then, the
following holds.

\begin{enumerate}[(i)]

\item $v_b\leq v_a.$

\item $r\leq v_a-v_b,$ where $r$ denotes the number of real roots of
the equation $P(x)=0$ located in the interval $(a,b).$

\item $v_a-v_b-r$ is either zero or an even number.

\end{enumerate}
\end{teo}

\section{Determinacy terminology}\label{determinacy terminology}
In order to simplify and shorten the statements we obtain in this
paper about determinacy of several models, our aim now is to introduce
a clearer notation on the determinacy question; the interested reader on semi--algebraic sets is referred to
\cite[Chapter 2]{RealAlgebraicGeometry} and the references given
therein for additional information.

\begin{df}\label{determinacy terms in action}
Let $n\geq 1$ be an integer, and let $S\subseteq\R^n$ be a semi--algebraic
set over $\R$ (that is, a subset of $\R^n$ satisfying a boolean
combination of polynomial equations and inequalities with
real coefficients). In practice, $S$ will be the space where the
parameters of our model lie.

\begin{enumerate}[(i)]

\item We say that our model is \textbf{unconditionally determined} if, for
any $(x_1,\ldots ,x_n)\in S,$ our model is determined.

\item We say that our model is \textbf{generically determined} if
there exists a semi--algebraic subset $S'\subset S$ such that our model is determined for any
$(x_1,\ldots ,x_n)\in S'.$

\end{enumerate}
Hereafter, we refer to the set $S$ as the \textbf{parameter space}, and
to $S'$ as the \textbf{determinacy region}.
\end{df}

\section{A dynamic linear system}\label{Gali model}

In this section we show that a set of non-restrictive assumptions on the structural parameters of the underlying economic model are sufficient for the uniqueness of the equilibrium. The case we consider was presented by Gal\'i in \cite[3.4.2]{GaliKeynesian2015}, but in contrast to the numerical methods in the original, here we also show the results analytically. In this particular case, one is interested in the analysis of the so-called timing structure: "cash-when-I'm-done" (CWID). Eventually, we are going to show that an exogenous money growth rule, under this specific setup, is always going to give us unconditional determinacy.
In the dynamic linear system considered in the paper, one wants to show that the matrix
\[
A:=\begin{pmatrix}[r] 1+\sigma\eta& 0& 0\\ -k& 1& 0\\ 0& -1& 1
\end{pmatrix}^{-1}\begin{pmatrix} \sigma\eta& \eta& 1\\ 0& \beta&
 0\\ 0& 0& 1\end{pmatrix},
\]
(where $k>0,$ $\sigma>0,$ $\eta>0$ and $\beta\in (0,1)$
are real numbers) has two eigenvalues inside the unit disk\footnote{Throughout this paper, by the \textbf{unit disk} we mean the set
$\{z\in\mathbb{C}:\ \lvert z\rvert <1\}.$} and
the remainder one is outside, because by \cite{BlanchardKhan1980}
this is equivalent to say that  the corresponding dynamic linear
system has a unique stationary solution; our goal in this section
is to show that this is true. Notice that, in this
case, our parameter space is
\[
S=\{(k,\sigma,\eta,\beta)\in\R^4:\ k>0,\ \sigma>0,\ \eta>0,\ 0< \beta<1\}.
\]
We deduce the unconditional determinacy of this model from the below
technical statement.

\begin{prop}\label{the conditions to apply Budan Fourier}
Let $A$ be a $3\times 3$ matrix with real entries such that its
characteristic polynomial is $P(x)=x^3-bx^2+cx-d,$ where we suppose that $b>1,$ $c>0,$ $d\in (0,1),$ $1-d<b-c,$ and $bc-d>0.$ Then, the
following assertions hold.

\begin{enumerate}[(i)]

\item All the real roots of $P$ are located in the interval $(0,b).$

\item $P$ has at least one real root in the interval $(1,b).$

\item If $P$ has two complex roots, then both are located in the unit disk.

\item If all the roots of $P$ are real, then $P$ has at least one real
root in the interval $(0,1).$

\item $P$ has a single root in the interval $(1,b).$

\end{enumerate}

\end{prop}

\begin{proof}
First of all, given $x\in [0,+\infty)$ notice that $P(-x)=-x^3-bx^2-cx-d<0,$
because $x\geq 0$ and $b>0,$ $c>0$ and $d>0$ by our assumptions; this
shows that $P$ has no roots in the interval $(-\infty, 0].$ Moreover, given $\mu\geq 0$
a real number, it follows that
\[
P(b+\mu)=\mu (b+\mu)^2+\mu c+(bc-d)>0;
\]
indeed, $P(b+\mu)>0$ because we know by assumption that
$P(b)=bc-d>0$ and the remainder terms of $P(b+\mu)$ are also non--negative. Summing
up, our calculations show that $P(x)<0$ for any $x\in (-\infty, 0],$ and
$P(x)>0$ for any $x\in [b,+\infty).$ These two facts show that
all the roots of $P$ are located in the interval $(0,b),$ and therefore
part (i) holds.

Now, the reader can easily check that $P(1)=1-b+c-d<0$ again
because of our assumptions; in this way, since
$P(1)<0$ and $P(b)>0$ Bolzano's theorem \cite[Proposition 1.2.4]{RealAlgebraicGeometry} guarantees the existence
of at least one real root of $P$ in the interval $(1,b),$ hence part
(ii) is also true. In what follows, we denote by $\lambda$ this
root, and let $\lambda_2,\lambda_3$ be the remainder roots of $P.$ Keeping
in mind this notation, one has that $\lambda\cdot\lambda_2\cdot\lambda_3=d$
if and only if $\lambda_2\cdot\lambda_3=d/\lambda.$ This equality shows that $\lambda_2\cdot\lambda_3<1$ because
$d\in (0,1),$ and $\lambda >1;$ now, we want to distinguish two cases.

On the one hand, suppose that $\lambda_2$ and $\lambda_3$
are complex numbers, then $\lambda_3=\overline{\lambda_2}$ (where
$\overline{(-)}$ denotes complex conjugation) and therefore
$\lvert\lambda_2\rvert^2=\lambda_2\cdot\overline{\lambda_2}
=\lambda_2\cdot\lambda_3<1\Longleftrightarrow\lvert\lambda_2\rvert<1,$
and this shows that both $\lambda_2$ and $\lambda_3$ are located in
the unit disk, as claimed.

On the other hand, if $\lambda_2$ and $\lambda_3$ are
real numbers, then either $\lambda_2\in (0,1)$ or
$\lambda_3\in (0,1)$ because $\lambda_2\cdot\lambda_3<1.$

Our final aim is to show that $P$ has a single root in the interval
$(1,b),$ and for this we plan to use the Budan--Fourier
Theorem (see Theorem \ref{Budan Fourier}); first of all, the sequence of $P$ and all its non--zero derivatives
(aka Fourier sequence) turns out to be
\[
P_{seq} (x):=(x^3-bx^2+cx-d,3x^2-2bx+c,6x-2b,6).
\]
Now, we evaluate this sequence respectively at $1$ and $b;$ namely,
\[
P_{seq} (1):=\left(1-b+c-d,3-2b+c,6-2b,6\right),\
P_{seq} (b):=\left(bc-d,b^2+c,4b,6\right)
\]
Let $v_1$ (respectively, $v_b$) be the number of signs variations
of $P_{seq} (1)$ (respectively, $P_{seq} (b)$), and notice that
$v_b=0$ because we know by our assumptions
that $P(b)=bc-d>0,$ $b^2+c>0$ and $4b>0.$

Remember that the Budan--Fourier Theorem says that the number of real
roots of $P$ located in the open interval $(1,b)$ is less or equal
than $v_1-v_b= v_1;$ our final aim will be to show that $v_1=1.$ We need
to consider four cases

Firstly, if $3-2b+c\leq 0$ and $6-2b\leq 0,$ then clearly $v_1=1.$ Secondly, if
$3-2b+c\leq 0$ and $6-2b\geq 0,$ then again $v_1=1.$ Thirdly, if
$3-2b+c\geq 0$ and $6-2b\geq 0,$ then once again $v_1=1$. Finally, assume
to reach a contradiction that $3-2b+c\geq 0$ and $6-2b\leq 0,$ so
$v_1=3.$ Notice that the inequality $6-2b\leq 0$ is equivalent
to say that $b\geq 3.$ On the other hand, the inequality $3-2b+c\geq 0$
is equivalent to $c\geq 2b-3,$ and this inequality implies, since
$b>c,$ that $b>2b-3,$ hence $b<3,$ a contradiction. Therefore, this
fourth case can not happen.

Summing up, we have finally checked that $v_1=1,$ which implies that
there is at most $1$ real root in the interval $(1,b)$ by the
Budan--Fourier Theorem, and since we already checked that
in this interval $P$ has at least one real root, we can finally conclude
that $P$ has a single root in the interval $(1,b),$ just what
we finally want to show.
\end{proof}
Now, building upon Proposition \ref{the conditions to apply Budan Fourier}, we
are ready to prove the main result of this section, keeping in
mind the notation we introduced at the very beginning.

\begin{teo}\label{some properties of matrix}
The following assertions hold.

\begin{enumerate}[(i)]

\item One has that
\[
A=\begin{pmatrix}[l] \frac{\sigma\eta}{1+\sigma\eta}& \frac{\eta}{1+\sigma\eta}
& \frac{1}{1+\sigma\eta}\\ \frac{k\sigma\eta}{1+\sigma\eta}& 
\frac{k\eta}{1+\sigma\eta}+\beta& \frac{k}{1+\sigma\eta}\\ 
\frac{k\sigma\eta}{1+\sigma\eta}& \frac{k\eta}{1+\sigma\eta}+\beta&
\frac{k}{1+\sigma\eta}+1
\end{pmatrix}
\]

\item The eigenvalues of $A$ are exactly the roots of the polynomial $P(x)=x^3-bx^2+cx-d,$
where
\[
b=\frac{\sigma\eta+k(1+\eta)}{1+\sigma\eta}+1+\beta,\ 
c=(1+\beta)\frac{\sigma\eta}{1+\sigma\eta}+
\frac{k\eta}{1+\sigma\eta}+\beta,\ d=\frac{\beta\sigma\eta}{1+\sigma\eta}.
\]

\item Our model is unconditionally determined.






\end{enumerate}

\end{teo}

\begin{proof}
First of all, part (i) is just an issue of inverting a matrix, and afterwards
a multiplication of matrices, and both steps are left to the
interested reader. Secondly, it is straightforward to check that $P(x)$ is
the characteristic polynomial of $A,$ hence part (ii) holds too.

To prove the unconditional determinacy of this
model, we only need to check that the
assumptions of Proposition \ref{the conditions to apply Budan Fourier}
hold; indeed, it is clear that $b>1,$ $c>0$ and $d\in (0,1).$ On the
other hand, notice that $P(1)=1-b+c-d=\frac{-k}{1+\sigma\eta}<0$ again
because $k$, $\sigma$ and $\eta$ are strictly positive; finally, one
also has that
\[
P(b)=bc-d=\left(\frac{\sigma\eta+k(1+\eta)}{1+\sigma\eta}+1+\beta\right)
\left((1+\beta)\frac{\sigma\eta}{1+\sigma\eta}+
\frac{k\eta}{1+\sigma\eta}+\beta\right)-\frac{\beta\sigma\eta}{1+\sigma\eta}>0.
\]
Summing up, we have checked that we are under the assumptions of
Proposition \ref{the conditions to apply Budan Fourier}, and therefore
this Proposition implies that our model is
unconditionally determined, just what we finally wanted to show.
\end{proof}

\begin{rk}
We want to single out that Proposition \ref{the conditions to apply Budan Fourier} does not cover
the model analyzed by Bullard and Mitra in \cite[Proposition 3 and
Appendix C]{BullardMitra2002} because, in their model, $d<0.$ Finally, the Routh--Hurwitz
criterion \cite[Theorem 1.1]{MeinsmaRH}, used to prove the stability of one of the models
studied by Gabaix (see \cite[Proposition 5.3]{GabaixNK} and \cite[Proposition 9.7]{GabaixNKOA}), only implies
in our case that all the eigenvalues of $P$ have positive real part, so it is not useful for
our purposes. In what follows (see Sections \ref{models with lagged data}
and \ref{behavioral model}), we analyze these models.
\end{rk}
We end this section by exhibiting some numerical examples to illustrate Theorem
\ref{some properties of matrix}; as we already pointed out in the
Introduction of this manuscript, the unjustified calculations in all
the examples we present in this paper were done with Matlab \cite{MATLAB:2015}. Remember
that our path here and in the remainder sections is, on the one
hand, to provide results that are obtained from theoretical analysis
and, on the other hand, to illustrate these results by means of
numerical examples.

\begin{ex}\label{example Gali result}
First of all, we calibrate our parameters in the following way: $\beta=0.99,$ $\sigma=0.5,$
$\eta=1.2$ and $k=0.3.$ In this case,
\[
A=\begin{pmatrix} 0.3750 & 0.7500 & 0.6250\\
0.1125 & 1.2150  & 0.1875\\
0.1125  &  1.2150  &  1.1875\end{pmatrix},
\]
its characteristic polynomial is
$x^3-2.7775x^2+1.9612x-0.3712$ (remember that Theorem \ref{some properties of matrix} guarantees
the existence of a unique eigenvalue in the interval
$(1,2.7775)$), and its eigenvalues are $0.3107,$ $0.6620$ and $1.8048.$ Therefore, in this
case, the unique eigenvalue of $A$ contained in $(1,2.7775)$ is $1.8048.$

Secondly, now we calibrate our parameters following
Woodford \cite{Woodford1999}; indeed, in this case we pick
$\beta=0.99,$ $\sigma=0.157,$
$\eta=1.2$ and $k=0.024.$ In this case,
\[
A=\begin{pmatrix} 0.1585  &  1.0098  &  0.8415\\
    0.0038  &  1.0142  &  0.0202\\
    0.0038  &  1.0142  &  1.0202\end{pmatrix},
\]
its characteristic polynomial is
$x^3-2.1930x^2+1.3297x-0.1569$ (remember that Theorem \ref{some properties of matrix} guarantees
the existence of a unique eigenvalue in the interval
$(1,2.1930)$), and its eigenvalues are $0.1547,$ $0.8634$ and $1.1749.$ Therefore, in this
case, the unique eigenvalue of $A$ contained in $(1,2.1930)$ is $1.1749.$

Finally, we calibrate our parameters following
Clarida, Gal\'i and Gertler \cite{ClaridaGaliGertler2000}; indeed, in this case we pick
$\beta=0.99,$ $\sigma=1,$
$\eta=1.2$ and $k=0.3.$ In this case,
\[
A=\begin{pmatrix} 0.5455  &  0.5455 &   0.4545\\
    0.1636  &  1.1536  &  0.1364\\
    0.1636  &  1.1536  &  1.1364\end{pmatrix},
\]
its characteristic polynomial is
$x^3-2.8355x^2+2.2391x-0.5400$ (remember that Theorem \ref{some properties of matrix} guarantees
the existence of a unique eigenvalue in the interval
$(1,2.8355)$), and its eigenvalues are $0.5455,$ $0.5784$ and $1.7116.$ Therefore, in this
case, the unique eigenvalue of $A$ contained in $(1,2.8355)$ is $1.7116.$
\end{ex}

\section{Some rules with lagged data}\label{models with lagged data}
Now, we would like to explore a more realistic version of the model. In what follows, policymakers are assumed to react to changes throughout particular policies, that were recorded in the past. In order to explore this fact, our next goal will be to recover and extend \cite[Proposition 3 and
Appendix C]{BullardMitra2002}; before doing so, we want to review the
following elementary Linear Algebra technical fact.

\begin{lm}\label{eigenvalues inside and outside}
Let $A$ be an invertible matrix with real entries. Then, $A$ has a
unique eigenvalue outside the unit disk if and only if $A^{-1}$
has a unique eigenvalue inside the unit disk.
\end{lm}

\begin{disc}\label{starting BM point}
For certain non--inertial lagged data rules \cite[pages 1118--1119]{BullardMitra2002}, the matrix relevant for uniqueness is the below one:
\[
B=\frac{1}{\varphi_x+k\varphi_{\pi}}\begin{pmatrix} 0& -\beta\varphi_{\pi}
& 1\\ 0& \beta\varphi_x & k\\ \sigma (\varphi_x+k\varphi_{\pi})& 
\varphi_x+(k+\beta\sigma)\varphi_{\pi}& -\sigma\end{pmatrix},
\]
where $k>0,$ $\sigma >0,$ $\beta\in (0,1),$ $\varphi_x\geq 0,$ $\varphi_{\pi}\geq 0,$ and
either $\varphi_x$ or $\varphi_{\pi}$ is strictly positive. By
Lemma \ref{eigenvalues inside and outside}, $B$ has two eigenvalues
inside the unit disk and one outside if and only if $B^{-1}$ has
one eigenvalue inside the unit disk and the remainder two ones outside. As
pointed out in \cite[Appendix C]{BullardMitra2002}, the characteristic
polynomial of $B^{-1}$ is $P(x)=x^3-bx^2+cx+d,$ where
\[
b=1+\frac{1}{\beta}+\frac{k}{\beta\sigma}>2,\ 
c=\frac{1}{\beta}-\frac{\varphi_x}{\sigma},\ 
d=\frac{\varphi_x+k\varphi_{\pi}}{\beta\sigma}.
\]
The reader will easily note that, in this model, our parameter space is
\begin{align*}
S=& \{(k,\sigma,\beta,\varphi_x, \varphi_{\pi})\in\R^5:\ k>0,
\ \sigma >0, 0<\beta<1,\ \varphi_x\geq 0,\ \varphi_{\pi}>0\}\\
& \cup\{(k,\sigma,\beta,\varphi_x, \varphi_{\pi})\in\R^5:\ k>0,
\ \sigma >0, 0<\beta<1,\ \varphi_x> 0,\ \varphi_{\pi}\geq 0\}.
\end{align*}
\end{disc}
Motivated by the content of Discussion \ref{starting BM point}, our next
goal will be to prove the following:

\begin{prop}\label{generalized BM model}
Let $P(x)=x^3-bx^2+cx+d\in\R [x],$ where $b>0,$ and $d>0.$ Then, the following assertions hold.

\begin{enumerate}[(i)]

\item $P$ has exactly one negative real root.

\item If $b>2,$ then $P$ has at least one root outside the unit disk.

\item $P$ has exactly one real root at $(-1,0)$ if and only
if $P(-1)<0.$

\item If $P(1)<0,$ then $P$ has exactly one real root at $(0,1).$

\item If $P(1)>0,$ and $b>2,$ then $P$ has a single real root at $(-\infty, 0),$ and
the other two roots are outside the unit disk.

\item If $P(-1)<0$ and $P(1)<0,$ then $P$ has exactly two real roots
in the interval $(-1,1)$ and the remainder real root is bigger strictly than $1.$

\item (Cf.\,\cite[Proposition 3]{BullardMitra2002}) Suppose that
$b>2.$ If $P(-1)<0$ and
$P(1)>0,$ then $P$ has exactly one root at $(-1,0),$ and the
remainder two ones are outside the unit disk.

\item Suppose that
$b>2.$ If $P(-1)>0$ and $P(1)<0,$ then $P$ has exactly one root
at $(0,1),$ and two roots whose real part is bigger than $1$ in absolute value.

\item Suppose that
$b>2.$ If $P(-1)>0$ and $P(1)>0,$ then $P$ has all its roots outside
the unit disk.

\item Suppose that
$b>2.$ $P$ has exactly one root at the unit disk and the remainder ones outside
if and only if one and only one of the following four conditions is satisfied.

\begin{itemize}

\item $P(-1)<0$ and $P(1)>0.$

\item $P(-1)=0,$ $P'(-1)\neq 0$ and $P(1)<0.$

\item $P(-1)>0$ and $P(1)<0.$

\item $P(-1)>0,$ $P(1)=0$ and $P'(1)\neq 0.$

\end{itemize}

\end{enumerate}

\end{prop}

\begin{proof}
First of all, the negative real roots of $P(x)$ are the positive real roots
of $Q(x)=P(-x)=-x^3-bx^2-cx+d;$ let $v_Q$ be the number of sign variations
of the coefficients of $Q.$ Independently of $c,$ one can see
that $v_Q=1,$ so Descartes' rule of signs \cite[Proposition 1.2.14]{RealAlgebraicGeometry}
implies that $P$ has exactly one negative real root.

Hereafter, let $\lambda_1,\lambda_2 ,\lambda_3$ be the roots of $P.$ Without
loss of generality, suppose that $\lambda_1<0,$ so it follows that
$\lambda_2+\lambda_3=b-\lambda_1>b>2.$ On the one hand, if $\lambda_2$ and $\lambda_3$ are real, then the
above upper inequality shows that either $\lambda_2>1$ or
$\lambda_3>1;$ on the other hand, if $\lambda_2$ and
$\lambda_3$ are complex conjugates, then again the above inequality
shows that their real part is bigger strictly than $1.$ Anyway, this
shows that $P$ has at least one root outside the unit disk.

Now, let $r$ be the number of real roots of $P$ at $(-1,0);$ notice that
\[
P_{seq} (x):=(x^3-bx^2+cx+d,3x^2-2bx+c,6x-2b,6).
\]
Now, we evaluate this sequence respectively at $-1$ and $0;$ namely,
\[
P_{seq} (-1):=\left(-1-b-c+d,3+2b+c,-6-2b,6\right),\
P_{seq} (0):=\left(d,c,-2b,6\right)
\]
Let $v_{-1}$ (respectively, $v_0$) be the number of signs variations
of $P_{seq} (-1)$ (respectively, $P_{seq} (0)$), and remember
that $r\leq v_{-1}-v_0$ by the Budan--Fourier Theorem. Moreover, notice
also that $v_0=2$ because $d>0,$ and $b>0.$ Next, there are four cases
to distinguish; first of all, if $-1-b-c+d<0$ and $3+2b+c<0,$ then
$c>-1-b+d,$ and therefore $0>3+2b+c>3+2b-1-b+d=2+b+d,$ a contradiction because
both $b$ and $d$ are strictly positive, hence this case can not happen. Secondly, if
$-1-b-c+d<0$ and $3+2b+c>0,$ then $v_{-1}=3$ and therefore, combining
Bolzano jointly with Budan--Fourier, we can guarantee that there is a
unique real root at $(-1,0).$ Thirdly, if $-1-b-c+d>0$ and $3+2b+c<0,$ then
$v_{-1}=2,$ hence no real roots at $(-1,0)$ by Budan--Fourier. Finally, if
if $-1-b-c+d>0$ and $3+2b+c>0,$ then once more $v_{-1}=2,$ so there
are no real roots at $(-1,0).$ Summing up, we have checked that $P$ has exactly one real root at $(-1,0)$ if and only
if $P(-1)<0,$ as claimed.

Next, we looked at the interval $(0,1)$ assuming that $P(1)<0;$ notice that
\[
P_{seq} (0):=\left(d,c,-2b,6\right),\
P_{seq} (1):=\left(1-b+c+d,3-2b+c,6-2b,6\right).
\]
Here, there are three cases to consider, keeping in mind that we are assuming
that $1-b+c+d<0;$ first of all, if $3-2b+c<0$ and either $6-2b<0$ or
$6-2b>0,$ then $v_1=1,$ so $v_0-v_1=2-1=1,$ and therefore Bolzano plus
Budan--Fourier ensure the existence of a unique real root at
$(0,1).$ Secondly, if $3-2b+c>0$ and $6-2b<0,$ then $v_1=3$ and thus
$v_0-v_1=-1,$ so this case can not happen because $0\leq v_0-v_1.$ Finally, if
$3-2b+c>0$ and $6-2b>0,$ then again $v_1=1,$ so $v_0-v_1=2-1=1,$ and therefore Bolzano plus
Budan--Fourier ensure the existence of a unique real root at
$(0,1).$ Summing up, we have checked that $P$ has exactly one real root at $(0,1)$ if $P(1)<0,$ as claimed.

Now, assume that $P(1)>0,$ and as above denote by $\lambda_1,\lambda_2 ,\lambda_3$ the roots of $P.$ Without
loss of generality, suppose that $\lambda_1<0,$ before we already saw that, if
$\lambda_2$ and $\lambda_3$ are complex, then both have real part strictly
bigger than $1,$ in particular they lie outside the unit disk. Moreover, we
also checked that, if $\lambda_2$ and $\lambda_3$ are real and
positive, then at least one of them is strictly bigger than $1,$ without
loss of generality suppose that $\lambda_2>1.$ If $\lambda_2
=\lambda_3,$ then we are done, so hereafter we assume that $\lambda_2
\neq\lambda_3,$ hence both are simple roots of $P.$ Suppose, to reach
a contradiction, that $\lambda_3\in (0,1);$ since $\lambda_3$ is a simple
root of $P$ and $P(0)>0,$ then there is $\varepsilon\in (0,1)$ such
that $\lambda_3\pm\varepsilon\in (0,1),$ $P(\lambda_3-\varepsilon)>0$ and
$P(\lambda_3+\varepsilon)<0.$ But this implies, since $P(1)>0,$ that
there is another real root at $(0,1)$ by Bolzano, a contradiction
by the foregoing. This shows that if $P(1)>0,$ then $P$ has a single real root at $(-\infty, 0),$ and
the other two roots are outside the unit disk.

Finally, notice that the remainder items (v)--(ix) are immediate consequence
of the previous ones, the proof is therefore completed.
\end{proof}

As immediate consequence of Proposition \ref{generalized BM model}, we
obtained our promised generalization of \cite[Proposition 3]{BullardMitra2002}, namely
the below:

\begin{teo}\label{the BM case}
Preserving the notations of Discussion \ref{starting BM point}, $B$ has
two eigenvalues inside the unit disk if and only if one and only one of the
following four conditions is satisfied.

\begin{enumerate}[(i)]

\item $k(\varphi_{\pi}-1)+(\varphi_x-2\sigma)(1+\beta)<0$ and
$k(\varphi_{\pi}-1)+\varphi_x(1-\beta)>0.$

\item $k(\varphi_{\pi}-1)+(\varphi_x-2\sigma)(1+\beta)=0,$
$\beta\varphi_x\neq \sigma(3+5\beta)+2k$ and 
$k(\varphi_{\pi}-1)+\varphi_x(1-\beta)<0.$

\item $k(\varphi_{\pi}-1)+(\varphi_x-2\sigma)(1+\beta)>0$ and
$k(\varphi_{\pi}-1)+\varphi_x(1-\beta)<0.$

\item $k(\varphi_{\pi}-1)+(\varphi_x-2\sigma)(1+\beta)>0,$
$k(\varphi_{\pi}-1)+\varphi_x(1-\beta)=0,$ and
$\beta\varphi_x\neq \sigma(\beta-1)-2k.$

\end{enumerate}
Therefore, in this case our model is generically determined, and the
determinacy region is
\begin{align*}
S'= &\{(k,\sigma,\beta,\varphi_x, \varphi_{\pi})\in S:\ 
k(\varphi_{\pi}-1)+(\varphi_x-2\sigma)(1+\beta)<0,
\ k(\varphi_{\pi}-1)+\varphi_x(1-\beta)>0\}\cup\\
& \{(k,\sigma,\beta,\varphi_x, \varphi_{\pi})\in S:\ 
k(\varphi_{\pi}-1)+(\varphi_x-2\sigma)(1+\beta)=0,
\beta\varphi_x\neq \sigma(3+5\beta)+2k,\\
& k(\varphi_{\pi}-1)+\varphi_x(1-\beta)<0\}\cup\\
& \{(k,\sigma,\beta,\varphi_x, \varphi_{\pi})\in S:\ 
k(\varphi_{\pi}-1)+(\varphi_x-2\sigma)(1+\beta)>0,
\ k(\varphi_{\pi}-1)+\varphi_x(1-\beta)<0\}\cup\\
& \{(k,\sigma,\beta,\varphi_x, \varphi_{\pi})\in S:\ 
k(\varphi_{\pi}-1)+(\varphi_x-2\sigma)(1+\beta)>0,\ 
\beta\varphi_x\neq\sigma(\beta-1)-2k,\\
& k(\varphi_{\pi}-1)+\varphi_x(1-\beta)=0\}.
\end{align*}
\end{teo}

\begin{rk}
Notice that, in Theorem \ref{the BM case}, the condition $k(\varphi_{\pi}-1)+\varphi_x(1-\beta)>0$
is what Woodford calls the \textbf{Taylor principle} (see
\cite{Woodford2001} and \cite{Woodford2003}); Theorem \ref{the BM case} shows, in particular, that
the Taylor principle is neither sufficient, nor necessary to guarantee determinacy. This fact was
already pointed out by Bullard and Mitra \cite[Propositions 1, 2 and 11]{BullardMitra2007}.
\end{rk}
Before moving to the next model, we want to consider the below:

\begin{ex}\label{no condition satisfied}
As we already proved, if all the conditions appearing in
Theorem \ref{the BM case} are not satisfied, then we can
not expect determinacy. As example, suppose that $\varphi_x=2.4,$
$\varphi_{\pi}=3.2,$ $\sigma=1,$ $\beta=0.99$ and $k=0.3;$ in this
case,
\[
B^{-1}=
\begin{pmatrix}[r] 1.3030  &  -1.0101 &   1\\
    -0.3030  &  1.0101  &  0\\
    2.4  &  3.2  &  0\end{pmatrix},
\]
its characteristic polynomial is
$P(x)=x^3-2.3131x^2-1.3899x+3.3939$ and its
eigenvalues are $-1.2003,$ $1.2482$ and $2.2653.$ Indeed, this
is because $P(-1)>0$ and $P(1)>0$ (cf.\, Proposition \ref{generalized BM model} (ix)).
\end{ex}

\begin{disc}\label{second BM model}
Now, we want to consider an inertial lagged data rule studied by Woodford \cite{Woodford2003} and Bullard and Mitra
\cite[page 1183]{BullardMitra2007}; in this specific model, the matrix
which is relevant to study determinacy is the below one:
\[
B=\begin{pmatrix}[r] 1+\frac{k\sigma}{\beta}& -\frac{\sigma}{\beta}
& \sigma\\ -\frac{k}{\beta}& \frac{1}{\beta} & 0\\ \varphi_x& 
\varphi_{\pi}& \varphi_{r}\end{pmatrix},
\]
where $k>0,$ $\sigma >0,$ $\beta\in (0,1),$ $\varphi_x\geq 0,$ $\varphi_{\pi}\geq 0,$ $\varphi_r\geq 0$ and
at least one among $\varphi_x,$ $\varphi_{\pi}$ and $\varphi_r$ is strictly positive. In this case, building upon
\cite[Appendix C, Proposition C.2]{Woodford2003}, Bullard and Mitra
\cite[Propositions 1, 2 and 11]{BullardMitra2007} gave
necessary and sufficient conditions to ensure determinacy; in this case, determinacy
holds if and only if $B$ has a single eigenvalue inside the unit disk. It is known
\cite[page 1198]{BullardMitra2007} that the characteristic polynomial of $B$ is
$P(x)=x^3-bx^2+cx+d,$ where
\[
b=1+\frac{1}{\beta}+\frac{k\sigma}{\beta}+\varphi_r>2,\ 
c=\frac{1}{\beta}+\left(1+\frac{1}{\beta}+\frac{k\sigma}{\beta}
\right)\varphi_r-\sigma\varphi_x,\ 
d=\frac{\sigma\left(\varphi_x+k\varphi_{\pi}-\sigma^{-1}\varphi_r\right)}{\beta}.
\]
Notice that, in this case, our parameter space is
\begin{align*}
S=& \{(k,\sigma,\beta,\varphi_x, \varphi_{\pi},\varphi_r)\in\R^6:\ k>0,
\ \sigma >0, 0<\beta<1,\ \varphi_x\geq 0,\ \varphi_{\pi}\geq 0,\ 
\varphi_r>0\},\\
& \cup\{(k,\sigma,\beta,\varphi_x, \varphi_{\pi},\varphi_r)\in\R^6:\ k>0,
\ \sigma >0, 0<\beta<1,\ \varphi_x\geq 0,\ \varphi_{\pi}> 0,
\varphi_r\geq 0\},\\
& \cup\{(k,\sigma,\beta,\varphi_x, \varphi_{\pi},\varphi_r)\in\R^6:\ k>0,
\ \sigma >0, 0<\beta<1,\ \varphi_x> 0,\ \varphi_{\pi}\geq 0,
\varphi_r\geq 0\}.
\end{align*}
\end{disc}
Once more, as immediate consequence of Proposition \ref{generalized BM model}, we
obtain the below:

\begin{teo}\label{the second BM case}
Preserving the notations of Discussion \ref{second BM model}, if
$\varphi_r<\sigma(k\varphi_{\pi}+\varphi_x),$ then $B$ has
a single eigenvalue inside the unit disk if and only if one of the
following conditions is satisfied.

\begin{enumerate}[(i)]

\item $k\sigma(\varphi_{\pi}-1)+(\sigma\varphi_x-2)(1+\beta)
-\varphi_r (2\beta+k\sigma+1)<0$ and
$k(\varphi_{\pi}-1+\varphi_r)+\varphi_x(1-\beta)>0.$

\item $k\sigma(\varphi_{\pi}-1)+(\sigma\varphi_x-2)(1+\beta)
-\varphi_r (2\beta+k\sigma+1)=0,$
$\sigma\beta\varphi_x\neq (3+5\beta+\varphi_r(1+3\beta+k\sigma))$ and
$k(\varphi_{\pi}-1+\varphi_r)+\varphi_x(1-\beta)<0.$

\item $k\sigma(\varphi_{\pi}-1)+(\sigma\varphi_x-2)(1+\beta)
-\varphi_r (2\beta+k\sigma+1)>0$ and
$k(\varphi_{\pi}-1+\varphi_r)+\varphi_x(1-\beta)<0.$

\item $k\sigma(\varphi_{\pi}-1)+(\sigma\varphi_x-2)(1+\beta)
-\varphi_r (2\beta+k\sigma+1)>0,$
$k(\varphi_{\pi}-1+\varphi_r)+\varphi_x(1-\beta)=0,$ and
$\sigma\beta\varphi_x\neq (\beta-2k\sigma-1+\varphi_r (1+k\sigma-\beta)).$

\end{enumerate}
Notice that, in this case, one has generic determinacy in
\begin{align*}
S'= & \{(k,\sigma,\beta,\varphi_x, \varphi_{\pi},\varphi_r)\in S:\ 
k\sigma(\varphi_{\pi}-1)+(\sigma\varphi_x-2)(1+\beta)
-\varphi_r (2\beta+k\sigma+1)<0,\\ & k(\varphi_{\pi}-1+\varphi_r)+\varphi_x(1-\beta)>0\}\\
& \cup\{(k,\sigma,\beta,\varphi_x, \varphi_{\pi},\varphi_r)\in S:\ 
k\sigma(\varphi_{\pi}-1)+(\sigma\varphi_x-2)(1+\beta)
-\varphi_r (2\beta+k\sigma+1)=0,\\ & k(\varphi_{\pi}-1+\varphi_r)+\varphi_x(1-\beta)<0,\
\sigma\beta\varphi_x\neq (3+5\beta+\varphi_r(1+3\beta+k\sigma))\}\\
& \cup\{(k,\sigma,\beta,\varphi_x, \varphi_{\pi},\varphi_r)\in S:\ 
k\sigma(\varphi_{\pi}-1)+(\sigma\varphi_x-2)(1+\beta)
-\varphi_r (2\beta+k\sigma+1)>0,\\ & k(\varphi_{\pi}-1+\varphi_r)+\varphi_x(1-\beta)<0,\}\\
& \cup\{(k,\sigma,\beta,\varphi_x, \varphi_{\pi},\varphi_r)\in S:\ 
k\sigma(\varphi_{\pi}-1)+(\sigma\varphi_x-2)(1+\beta)
-\varphi_r (2\beta+k\sigma+1)>0,\\ & k(\varphi_{\pi}-1+\varphi_r)+\varphi_x(1-\beta)=0,\
\sigma\beta\varphi_x\neq (\beta-2k\sigma-1+\varphi_r (1+k\sigma-\beta))\}.
\end{align*}

\end{teo}

\begin{rk}
Notice that both Theorem \ref{the BM case} and Theorem \ref{the second BM case} deal
even with non--generic boundary cases; in case of Theorem \ref{the second BM case}, Woodford
already observed \cite[footnote of page 672]{Woodford2003} that his
determinacy conditions are sufficient but not generically necessary, whereas the ones
we are providing in our results work with full generality. Finally, observe
that the determinacy conditions obtained in Theorem
\ref{the second BM case} only work, roughly speaking, for bounded
values of inertia, whereas Bullard and Mitra's ones
\cite[Proposition 2]{BullardMitra2007} work for unbounded
inertia.
\end{rk}

We end the discussion of this model with the below:

\begin{ex}\label{even low inertia yields no determinacy}
We want to single out that, of course, the assumption
$\varphi_r<\sigma(k\varphi_{\pi}+\varphi_x)$ is not solely
enough to ensure determinacy. As example, suppose that
$\beta=0.99,$ $k=0.3,$ $\sigma=1,$ $\varphi_x=4.3,$ $\varphi_{\pi}
=1.82$ and $\varphi_r=0.5;$ in this
case,
\[
B=
\begin{pmatrix}[r] 1.3030  &  -1.0101 &   1\\
    -0.3030  &  1.0101  &  0\\
    2.4  &  3.2  &  0.5\end{pmatrix},
\]
its characteristic polynomial is
$P(x)=x^3-2.8131x^2-2.1333x+4.3899$ and its
eigenvalues are $-1.3204,$ $1.0937$ and $3.0399.$ Indeed, this
is because $P(-1)=2.7101>0$ and $P(1)=0.4434>0$ (cf.\, Proposition \ref{generalized BM model} (ix)).
\end{ex}

\section{Studying a behavioral New Keynesian model}\label{behavioral model}
As consequence of the Routh-Hurwitz criterion \cite[Theorem 1.1]{MeinsmaRH}, Gabaix
(see \cite[Proposition 5.3]{GabaixNK} and \cite[Proposition 9.7]{GabaixNKOA})
obtained the below:

\begin{prop}[Gabaix]\label{stability of Gabaix model}
Let $P(x)=x^3-bx^2+cx-d\in\R [x],$ where $b>0,$ $c>0,$ $d>0$ and
$P(1)\neq 0.$ Then, the following statements are equivalent.

\begin{enumerate}[(i)]

\item $P$ has exactly one root at $(0,1)$ and
the remainder ones are outside the complex unit disk.

\item The sequence $(e_3,e_2,(e_2 e_1-e_3e_0)/e_2,e_0)$
contains exactly two sign changes, where $e_3=1-b+c-d,$ $e_2=3-b
-c+3d,$ $e_1=3+b+c-3d$ and $e_0=1+b+c+d.$

\item Either $e_2\leq 0$ or $e_2 e_1-e_3e_0\leq 0.$

\end{enumerate}

\end{prop}
Let us briefly review what was the original motivation for Gabaix
to look at Proposition \ref{stability of Gabaix model}; indeed, building
upon a Taylor stability criterion which includes behavioral agents
\cite[Proposition 3.1]{GabaixNK}, Gabaix introduced a behavioral
New Keynesian Model with backward looking terms (see
\cite[Proposition 5.3]{GabaixNK} and \cite[Proposition 9.7]{GabaixNKOA}). In
this extended model, the relevant matrix to ensure determinacy is the
below one:
\[
B=\begin{pmatrix} \frac{\sigma\phi_x\beta^f+\beta^f+k\sigma}{M\beta^f}&
\frac{\sigma (\beta\phi_{\pi}-\alpha^f\eta\rho\chi-1)}{M\beta^f}&
\frac{\alpha^f \sigma((\eta-1)\rho+1)}{M\beta^f}\\
-\frac{k}{\beta^f}& \frac{\alpha^f\eta\rho\chi+1}{\beta^f}
& \frac{\alpha^f(-\eta\rho+\rho-1)}{\beta^f}\\
0& \eta\chi& 1-\eta\end{pmatrix}.
\]
In this case, since in this model there is a single predetermined
variable and the remainder two ones are jump variables, again
using \cite{BlanchardKhan1980}, determinacy holds if and only if
$B$ has a single real eigenvalue less than $1$ in absolute value, and
the remainder two ones are complex number with modulus greater
than one.

We also want to single out that, among all the parameters involved in the
above matrix, both $\phi_x$ and $\phi_{\pi}$ are non--negative, and
both $M,M^f\in [0,1]$ represent a degree of behavioralism in the
models studied by Gabaix, as already observed by Cochrane in
\cite{Cochranecomments}.

Going back to Proposition \ref{stability of Gabaix model}, notice that the expression $e_2 e_1-e_3e_0$ is quadratic in terms
of the coefficients of our polynomial; our next goal will be to
provide a sufficient condition to guarantee stability that only involves
a linear expression in the coefficients of the polynomial; namely:

\begin{prop}\label{sufficient linear condition}
Let $P(x)=x^3-bx^2+cx-d\in\R [x],$ where $b>0,$ $c>0,$ $d>0$ and
$P(1)\neq 0.$ Then, the following assertions hold.

\begin{enumerate}[(i)]

\item All the real roots of $P$ are contained in the interval
$(0,1+M),$ where $M=\max\{b,c,d\}.$

\item If $P$ has only one real root at $(0,1),$ then $P(1)\geq 0.$

\item If either $3-2b+c\leq 0$ or $b\geq 3,$ and $P(1)>0,$ then $P$ has a single
real root at $(0,1).$

\item If $b-c>0$ and $P(1)>0,$ then $P$ has a single real root at $(0,1).$

\end{enumerate}

\end{prop}

\begin{proof}
Let $x_0\in [0,+\infty),$ and notice that $P(-x_0)=
-x_0^3-bx_0^2-cx_0-d<0;$ this shows that all the real roots
of $P$ are strictly positive. The fact that all of them are less
than $1+\max\{b,c,d\}$ is just by the classical Cauchy bound
\cite[Theorem 1]{HirstMacey1997}; this proves part (i).

Now, assume, to reach a contradiction, that $P(1)<0;$ keeping in mind
that $P(0)=-d<0,$ we have to distinguish two cases. On the one hand, if
$P(x_0)>0$ for some $x_0\in (0,1),$ then Bolzano's Theorem implies that
there are at least two real roots at $(0,1)$ (indeeed, because
$P(0)<0,$ $P(x_0)>0$ and $P(1)<0$), so we get a contradiction. On the
other hand, if $P(x_0)\leq 0$ for all $x_0\in (0,1)$ and
$P(\lambda)=0$ for some $\lambda\in (0,1),$ then $\lambda$ has to be
of multiplicity two, and this is again a contradiction.

Finally, notice that parts (iii) and (iv) were already shown in the
course of the proofs of Propositions \ref{some properties of matrix} and \ref{generalized BM model};
the proof is therefore completed.
\end{proof}

\begin{disc}\label{determinacy region at Gabaix model}
Our plan here is to use Proposition \ref{sufficient linear condition}
to partially describe the sets where determinacy holds in Gabaix model; indeed, remember
that in his model, the relevant matrix to ensure determinacy is the
below one:
\[
B=\begin{pmatrix} \frac{\sigma\phi_x\beta^f+\beta^f+k\sigma}{M\beta^f}&
\frac{\sigma (\beta\phi_{\pi}-\alpha^f\eta\rho\chi-1)}{M\beta^f}&
\frac{\alpha^f \sigma((\eta-1)\rho+1)}{M\beta^f}\\
-\frac{k}{\beta^f}& \frac{\alpha^f\eta\rho\chi+1}{\beta^f}
& \frac{\alpha^f(-\eta\rho+\rho-1)}{\beta^f}\\
0& \eta\chi& 1-\eta\end{pmatrix}.
\]
One can check, using the expression of the characteristic polynomial of $B$
written down by Gabaix \cite[page 64]{GabaixNKOA} jointly with the value
of this polynomial evaluated at $1$ \cite[page 66]{GabaixNKOA} the below facts.

First of all, part (ii) of Proposition \ref{sufficient linear condition}
shows that the determinacy region must be contained inside
\begin{align*}
\widetilde{S}:= &\{(k,\sigma,\alpha,\alpha^f,\beta,\beta^f,M,M^f,\eta,
\rho,\chi,\phi_x,\phi_{\pi})\in\R^{13}: \\
& (1-\beta^f-\alpha\chi(1-\rho))(1-M+\sigma\phi_x)+k\sigma (\phi_{\pi}-1)\geq 0\}.
\end{align*}
Secondly, part (iv) of Proposition \ref{sufficient linear condition}
shows that the determinacy region must contain the subset
of $\widetilde{S}$ given by the inequality
\begin{align*}
& \left(\frac{\sigma((\beta^f-1)+\beta (\eta-1)-\eta\alpha^f\rho\chi)}{M\beta^f}\right)
-\left(\frac{k\sigma}{M\beta^f}\right)\phi_{\pi}\\
& +\frac{(\eta-1)(k\sigma+1+\beta+M-M\beta^f)+\eta(\alpha^f\chi(\rho(M-1)-M)-1)
+M+\beta^f+k\sigma}{M\beta^f}>0.
\end{align*}
Finally, part (iii) of Proposition \ref{sufficient linear condition}
shows that the determinacy region must contain, on the one
hand, the subset of $\widetilde{S}$ given by the inequality
\[
\left(\left(\frac{\sigma}{M}\right)\phi_x+\frac{(1-\eta)M\beta^f+\eta\alpha^f\rho
\chi M+M+\beta^f+k\sigma}{M\beta^f}\right)\geq 3,
\]
and, on the other hand, the subset of $\widetilde{S}$ given by the inequality
\begin{align*}
& \left(\frac{\sigma((\beta^f-1)+\beta (\eta-1)-\eta\alpha^f\rho\chi)}{M\beta^f}\right)
-\left(\frac{k\sigma}{M\beta^f}\right)\phi_{\pi}\\
& +\frac{(\eta-1)(k\sigma+1+\beta+M-M\beta^f)+\eta(\alpha^f\chi(\rho(M-1)-M)-1)
+M+\beta^f+k\sigma}{M\beta^f}\\
& +\left(\left(\frac{\sigma}{M}\right)\phi_x+\frac{(1-\eta)M\beta^f+\eta\alpha^f\rho
\chi M+M+\beta^f+k\sigma}{M\beta^f}\right)\geq 3.
\end{align*}
\end{disc}

One can easily check that Proposition \ref{stability of Gabaix model} and Proposition
\ref{sufficient linear condition} can be applied to obtain necessary
and sufficient (respectively, sufficient) conditions to guarantee
the determinacy of the model studied by Bullard and Mitra in
\cite[page 1185]{BullardMitra2007}; here, we only write down the
sufficient conditions of determinacy given by Proposition
\ref{sufficient linear condition} in their specific model
(cf. \cite[Propositions 3 and 4]{BullardMitra2007}).

\begin{teo}\label{Bullard and Mitra other sufficient conditions}
The matrix
\[
\frac{1}{1-\varphi_x \sigma}\begin{pmatrix}1-\beta^{-1}k\sigma
(\varphi_{\pi}-1)& \beta^{-1}\sigma (\varphi_{\pi}-1)& 
\sigma\varphi_r\\ -k\beta^{-1}(1-\varphi_x \sigma)&
\beta^{-1}(1-\varphi_x \sigma)& 0\\
\varphi_x (1+\beta^{-1}k\sigma)-\beta^{-1}k\varphi_{\pi}&
\beta^{-1}(\varphi_{\pi}-\varphi_x \sigma)& \varphi_r
\end{pmatrix},
\]
(where $k>0,$ $\sigma>0,$ $\beta>0,$ $\varphi_x\geq 0,$
$\varphi_{\pi}\geq 0,$ $\varphi_r\geq 0,$ and at least one among
$\varphi_x,\varphi_{\pi},\varphi_r$ strictly positive) has exactly one
eigenvalue at $(0,1)$ if $\varphi_x<\sigma^{-1},$ $\varphi_{\pi}\leq 1,$
\[
(1-\varphi_x \sigma)\left(
\beta(1-\beta)-\varphi_r(1+(1+\varphi_r)(\beta+k\sigma))\right)
+\beta\left(\beta+\varphi_r (\beta+1)+k\sigma (1-\varphi_{\pi})\right)<0,
\]
and, in addition, at least one of the below inequalities holds:
\[
\begin{cases} (1-\varphi_x \sigma)\left(\beta(2-3\beta)
-\varphi_r (1+(1+\varphi_r)(\beta+k\sigma))\right)
+\beta\left(2\beta(1+\varphi_r)+2k\sigma(1-\varphi_{\pi})\right)\geq 0,\\
(1-\varphi_x \sigma)(1-3\beta)+\beta(1+\varphi_r)
+k\sigma (1-\varphi_{\pi})\geq 0,\\
(1-\varphi_x\sigma)\left(\beta-\varphi_r-(1+\varphi_r)
(\beta+k\sigma)\right)+\beta\left(\beta (1+\varphi_r)
+k\sigma (1-\varphi_{\pi})\right)>0.\end{cases}
\]
In this case, our parameter space is
\begin{align*}
S=& \{(k,\sigma,\beta,\varphi_x, \varphi_{\pi},\varphi_r)\in\R^6:\ k>0,
\ \sigma >0, 0<\beta<1,\ \varphi_x\geq 0,\ \varphi_{\pi}\geq 0,\ 
\varphi_r>0\},\\
& \cup\{(k,\sigma,\beta,\varphi_x, \varphi_{\pi},\varphi_r)\in\R^6:\ k>0,
\ \sigma >0, 0<\beta<1,\ \varphi_x\geq 0,\ \varphi_{\pi}> 0,
\varphi_r\geq 0\},\\
& \cup\{(k,\sigma,\beta,\varphi_x, \varphi_{\pi},\varphi_r)\in\R^6:\ k>0,
\ \sigma >0, 0<\beta<1,\ \varphi_x> 0,\ \varphi_{\pi}\geq 0,
\varphi_r\geq 0\},
\end{align*}
and generic determinacy holds in the below subset:
\begin{align*}
S'=& \{(k,\sigma,\beta,\varphi_x, \varphi_{\pi},\varphi_r)\in S:
\ \sigma\varphi_x<1,\ \varphi_{\pi}\leq 1,\\
& (1-\varphi_x \sigma)\left(
\beta(1-\beta)-\varphi_r(1+(1+\varphi_r)(\beta+k\sigma))\right)
+\beta\left(\beta+\varphi_r (\beta+1)+k\sigma (1-\varphi_{\pi})\right)<0,\\
& (1-\varphi_x \sigma)\left(\beta(2-3\beta)
-\varphi_r (1+(1+\varphi_r)(\beta+k\sigma))\right)
+\beta\left(2\beta(1+\varphi_r)+2k\sigma(1-\varphi_{\pi})\right)\geq 0\},\\
& \cup
\{(k,\sigma,\beta,\varphi_x, \varphi_{\pi},\varphi_r)\in S:
\ \sigma\varphi_x<1,\ \varphi_{\pi}\leq 1,\\
& (1-\varphi_x \sigma)\left(
\beta(1-\beta)-\varphi_r(1+(1+\varphi_r)(\beta+k\sigma))\right)
+\beta\left(\beta+\varphi_r (\beta+1)+k\sigma (1-\varphi_{\pi})\right)<0,\\
& (1-\varphi_x \sigma)(1-3\beta)+\beta(1+\varphi_r)
+k\sigma (1-\varphi_{\pi})\geq 0\}\\
& \cup\{(k,\sigma,\beta,\varphi_x, \varphi_{\pi},\varphi_r)\in S:
\ \sigma\varphi_x<1,\ \varphi_{\pi}\leq 1,\\
& (1-\varphi_x \sigma)\left(
\beta(1-\beta)-\varphi_r(1+(1+\varphi_r)(\beta+k\sigma))\right)
+\beta\left(\beta+\varphi_r (\beta+1)+k\sigma (1-\varphi_{\pi})\right)<0,\\
& (1-\varphi_x\sigma)\left(\beta-\varphi_r-(1+\varphi_r)
(\beta+k\sigma)\right)+\beta\left(\beta (1+\varphi_r)
+k\sigma (1-\varphi_{\pi})\right)>0\}.
\end{align*}
\end{teo}
We also want to illustrate this model with the below:

\begin{ex}\label{example of the third Bullard Mitra model}
Suppose that
$\beta=0.99,$ $k=0.3,$ $\sigma=1,$ $\varphi_x=4.3,$ $\varphi_{\pi}
=1.82$ and $\varphi_r=0.5;$ in this
case,
\[
\begin{pmatrix}[r] -0.2277  &  -0.2510 & -0.1515\\
    -0.3030  &  1.0101  &  0\\
    -1.5308  &  0.7591  &  -0.1515\end{pmatrix},
\]
its characteristic polynomial is
$P(x)=x^3-0.6309x^2-0.6566x+0.1530$ and its
eigenvalues are $-0.6758,$ $0.2057$ and $1.1010.$ Indeed, this
is because $P(-1)=-0.8212<0$ and $P(1)=-0.1344<0$ (cf.\, Proposition \ref{generalized BM model} (vi)). Notice
that, in this case, $0.6309<2.$
\end{ex}

We want to conclude this section with the below:

\begin{ex}\label{non--negativity at 1 is important}
We want to single out that the assumption
$P(1)\geq 0$ is necessary (but not sufficient) to ensure determinacy. As example, suppose that
$\phi_x=1,$ $\phi_{\pi}=2,$ $\alpha=0.5,$ $\rho=0.35,$ $\beta=0.99,$
$\eta=0.05,$ $\chi=0.3,$ $m=0.85,$ $\sigma=0.2,$ and
$k=0.053.$ In this
case,
\[
B=
\begin{pmatrix}[r] 1.4244  &   0.2323 &   1.1488\\
    -0.0535  &  1.0177  &  -0.3371\\
    0  &  0.0150  &  0.9500\end{pmatrix},
\]
its characteristic polynomial is
$P(x)=x^3-3.3920x^2+3.7870x-1.3952$ and its
eigenvalues are $1.3861,$ $1.0030 + 0.0240i$ and $1.0030 -0.0240i.$ Indeed, this
is because $P(1)=-2.2650e-04<0.$  Notice that, even in this case, one
has that $b-c<0.$
\end{ex}

\section{Potential limitations}\label{potential limitations}
So far in this paper, the use of the Budan--Fourier Theorem to address the determinacy issue has
been restricted to models where the characteristic equation one has
to deal with is of degree three, so one natural question to ask is the
potential use of this result to tackle models where the characteristic
equation has higher degree; our goal in this section will be to briefly
explain what might be the potential limitations of doing so. We illustrate
it in what follows.

Indeed, we consider one of the models studied by Bhattarai, Lee and Park
in \cite{BhattaraiLeePark2014}; in such a model, the parameter space is
\begin{align*}
S=\{& (\beta,\eta,\gamma,\rho_R,k,\varphi,\phi_Y,\phi_{\pi})\in\R^8:\ 
\beta\in (0,1),\ \eta\in [0,1),\ \rho_R\in [0,1),\ \gamma\in [0,1],\\
& k>0,\ \varphi>0,\ \phi_Y>0,\ \phi_{\pi}>0\},
\end{align*}
and one has to look at the polynomial $P(x)=a_5x^5-a_4x^4+a_3x^3
-a_2x^2+a_1x-a_0,$ where
\begin{align*}
& a_5=\beta,\ a_4=\beta+1+\beta(\eta+\gamma+\rho_R)+
(1-\eta)k\left(\varphi+\frac{1}{1-\eta}+(1-\rho_R)\phi_Yk^{-1}\beta\right),\\
& a_3=1+(\beta+1)(\eta+\gamma+\rho_R)+\beta (\eta\gamma+\eta\rho_R+\gamma
\rho_R)
\\ &+(1-\eta)(1-\rho_R)k\left(\left(\phi_{\pi}+
\frac{\rho_R}{1-\rho_R}\right)\left(\varphi+
\frac{1}{1-\eta}\right)+(1+\beta\gamma)\phi_Y k^{-1}
+\frac{1}{1-\rho_R}\left(\frac{\eta}{1-\eta}\right)
\right),\\
& a_2=(\eta+\gamma+\rho_R)+(\beta+1)(\eta\gamma+\eta\rho_R+\gamma
\rho_R)+\beta\eta\gamma\rho_R\\
& +(1-\eta)(1-\rho_R)k\left(\left(\phi_{\pi}+
\frac{\rho_R}{1-\rho_R}\right)\left(\varphi+
\frac{1}{1-\eta}\right)+\phi_Y\gamma k^{-1}\right),\\
& a_1=\eta\gamma+\rho_R (\eta+\gamma+\eta\gamma+\beta\eta\gamma),
\ a_0=\eta\gamma\rho_R.
\end{align*}
Since $P(-x)\leq 0$ for any $x\in [0,+\infty),$ all the real roots of
$P$ need to be strictly positive; moreover, since the degree of $P$ is
odd, Bolzano's Theorem guarantees that $P$ has at least one real
root. On the other hand, since in this case determinacy holds if and
only if $P$ has exactly three roots inside the unit disk, at least
one of them has to be real, hence contained at (0,1) (indeed, otherwise
$P$ would have $0,$ $2$ or $4$ roots inside the unit disk); summing up, as
observed in \cite{BhattaraiLeePark2014}, a necessary condition for
determinacy is $P(1)>0.$ Therefore, in this case, the determinacy region
must be contained inside
\[
\left\{(\beta,\eta,\gamma,\rho_R,k,\varphi,\phi_Y,\phi_{\pi})\in
S:\ \phi_{\pi}+\frac{(1-\gamma)(1-\beta)}{k(\varphi+1)}\phi_Y >1
\right\}.
\]
Our next goal will be to show that the Budan--Fourier Theorem provide
other necessary conditions for determinacy; indeed, let $v_1$ be the
number of sign variations of
\begin{align*}
P_{seq} (1)=& (a_5-a_4+a_3-a_2+a_1-a_0,5a_5-4a_4+3a_3-2a_2+a_1,
2(10a_5-6a_4+3a_3-a_2),\\ & 6(10a_5-4a_4+a_3),24(5a_5-a_4),120a_5).
\end{align*}
Moreover, since one can easily check that $v_0=5,$ the Budan--Fourier Theorem
tells us that the number $r$ of real roots of $P$ at (0,1) is less
or equal than $5-v_1,$ and that $5-v_1-r$ is either zero, two or four. In
this case, since determinacy holds if and only if there are exactly three
roots inside the unit circle, it follows that $r$ can only be either
one or three, which implies that $5-v_1$ must be zero, one
or three, which is equivalent to say that $v_1$ can only be zero, two
or four. Summing up, this shows that the determinacy region must be contained
inside
\begin{align*}
\widetilde{S}:=& \left\{(\beta,\eta,\gamma,\rho_R,k,\varphi,\phi_Y,\phi_{\pi})\in
S:\ \phi_{\pi}+\frac{(1-\gamma)(1-\beta)}{k(\varphi+1)}\phi_Y >1,\
v_1=0\right\}\\
& \cup\left\{(\beta,\eta,\gamma,\rho_R,k,\varphi,\phi_Y,\phi_{\pi})\in
S:\ \phi_{\pi}+\frac{(1-\gamma)(1-\beta)}{k(\varphi+1)}\phi_Y >1,\
v_1=2\right\}\\
& \cup\left\{(\beta,\eta,\gamma,\rho_R,k,\varphi,\phi_Y,\phi_{\pi})\in
S:\ \phi_{\pi}+\frac{(1-\gamma)(1-\beta)}{k(\varphi+1)}\phi_Y >1,\
v_1=4\right\}.
\end{align*}
The reader will easily note that, while the necessary and sufficient conditions obtained
in \cite[Proposition of page 222]{BhattaraiLeePark2014} by means of a stronger
version of the Rouch\'e Theorem \cite[Theorem 2]{Lloyd1979} require
the evaluation of transcendental functions, the necessary conditions
we give through Budan--Fourier just involve polynomial evaluation.

\section*{Conclusion}
By means of the Budan--Fourier Theorem \cite[Theorem 1]{Akritas1982}, we
have shown in a completely analytical way the existence and uniqueness
of real roots for several linear systems of equations arising from
New Keynesian models; indeed, we have done so, first of all, for a
model when the money supply follows an exogenous path \cite[3.4.2]{GaliKeynesian2015}, secondly
when a monetary authority responds to lagged values of output
(see \cite[Proposition 3 and Appendix C]{BullardMitra2002} and \cite[Propositions 1, 2 and 11]{BullardMitra2007}), and finally when agents do not fully
understand future policies (see \cite[Proposition 5.3]{GabaixNK} and \cite[Proposition 9.7]{GabaixNKOA}). We
also pinpoint the potential limitations of these methods to tackle
models where the characteristic equation is of high degree. It is well known that, when
the characteristic equation is of degree two, there are several more
elementary ways to tackle this issue; for instance, Chatelain and
Ralf \cite[Proposition 1]{ChatelainRalf2018} use the fact that, when
the characteristic equation is of degree two, the eigenvalues
are non--linear functions of the trace and the determinant of the corresponding
matrix \cite[pages 63--67]{Azariadis1993}.

One thing the reader may ask is why we have only used the Budan--Fourier
Theorem to estimate the number of real roots of a polynomial in a given
interval; this is because the models studied in this paper involved
between four and thirteen parameters, and from our perspective it is not
obvious, neither to evaluate polynomials in the whole interval we
are interested in, nor to make too many manipulations with
them. This prevented us to employ other techniques, like Sturm
sequences \cite[Corollary 1.2.10]{RealAlgebraicGeometry}, to tackle this complicated
issue; of course, obtaining complete necessary and sufficient determinacy conditions require, in
general, not only to look at the real roots, but also at the complex ones. This
explains why in recent works (e.g. \cite{Lubik2007}, \cite{GabaixNK}, \cite{BhattaraiLeePark2014}) authors dealing
with the determinacy issue use more sophisticated tools, like the
Routh--Hurwitz criterion \cite[Theorem 1.1]{MeinsmaRH}, the Schur--Cohn
criterion (see \cite[page 198, Th. (43,1)]{Mardenbook} or
\cite[page 27, 5.3.]{LaSallebook}) or the
Rouch\'e Theorem \cite[Theorem 2]{Lloyd1979}.

We would like to single out that our motivation comes from the issue of indeterminacy
of the rational expectations equilibrium that complicated the
conduct of monetary policy, and also with the multiple equilibria
puzzle problem arising from New Keynesian models; we hope the techniques
used throughout this paper can
help to tackle, not only these issues, but also others appearing
in models different from the ones considered here. Once again, as we already mentioned in the Introduction, we repeat that all the techniques and most of the examples presented here are not new, what might be original in this manuscript is the organization of the material and the emphasis, hoping that will be potentially useful for researchers working in this subject. The list of references at the end gives an indication of the provenance of the fundamental ideas and techniques, and might suggest directions for additional research. 

\section*{Funding}
This work was partially supported by Spanish Ministerio de Econom\'ia y Competitividad [MTM2019-104844GB-I00].

\section*{Acknowledgements}
The authors would like to thank Davide Debortoli, Jordi Gal\'i, Thomas Lubik, Lluc Puig and Jes\'us Fern\'andez-Villaverde for useful and fruitful advices and feedback about the contents of this paper.



\bibliographystyle{alpha}
\bibliography{AFBoixReferences}

\end{document}